\documentclass[12pt]{article}
\usepackage{amsmath, amsthm, amsfonts, amssymb, amsxtra}
\usepackage{graphicx,psfrag,epsf}
\usepackage{enumerate}
\usepackage{natbib}
\usepackage{url} % not crucial - just used below for the URL
\usepackage{bm, bbm}
\usepackage{booktabs}
\usepackage[pagebackref=false]{hyperref} % pagebackref to true to check how many
\usepackage{xcolor}
\usepackage{subcaption}
\usepackage{tikz}
\usepackage{physics} % for \dd differentiation operator
\usepackage[]{lineno}
\usepackage[margin = 1in]{geometry}
\usepackage{setspace}

\newtheorem{proposition}{Proposition}[section]
\newtheorem{theorem}{Theorem}[section]
\newtheorem{corollary}{Corollary}[theorem]
\newtheorem{lemma}[theorem]{Lemma}

\theoremstyle{definition}
\newtheorem{definition}[theorem]{Definition}
\newtheorem{remark}{Remark}

\newcommand{\B}{\mathrm{B}}

\newcommand{\E}{\mathbbm{E}}
\graphicspath{ {img/}}
\hypersetup{
    colorlinks = true,
    linkcolor  = blue,
    urlcolor   = blue,
    citecolor  = blue,
    pdftitle   = {Isotropic covariance functions for sets} 
  }

\begin{document}

\title{Isotropic Covariance Functions on Sets: A Center-Radius Alternative to
  the Hausdorff Distance}

\author{
  Lucas da Cunha Godoy\thanks{Department of Statistics, Pontificia
  Universidad Cat\'{o}lica de Chile, Santiago, Chile} \and
  Marcos Oliveira Prates\thanks{Department of Statistics, Universidade Federal
    de Minas Gerais, Brasil} \and
  Fernando Andr\'es Quintana\thanks{Department of Statistics, Pontificia
  Universidad Cat\'{o}lica de Chile, Santiago, Chile} \and
  Jun Yan\thanks{Department of Statistics, University of
    Connecticut, USA}
}

\maketitle
\begin{abstract}
  Spatial data often combine point-referenced and areal observations. The
  traditional aggregated model averages a point-level Gaussian process over each
  areal unit, which typically requires numerical integration. Isotropic
  covariance functions defined directly on sets avoid this integration, and a
  conditionally negative definite~(CND) set distance makes broad families of
  them valid.  However, the Hausdorff distance can fail to be CND even when the
  underlying metric is CND. To build such covariance functions, we propose the
  center-radius~(CR) pseudometric for bounded sets with unique Chebyshev
  centers. The CR distance equals the distance between the centers of the
  minimum enclosing balls of two sets plus the absolute difference of their
  radii. In Euclidean spaces, this sum equals the Hausdorff distance between the
  two balls. We prove that the CR distance is CND whenever the underlying metric
  is CND. The square-root CR distance therefore admits a Hilbert-space
  representation, yielding valid set-indexed covariance functions from the
  Mat\'{e}rn, powered exponential, and other families. In numerical benchmarks,
  the CR and Hausdorff distances between sets have correlations of at least
  0.99. The CR distance is also~33 to~647 times faster to compute. We applied
  the CR distance to fuse \textit{in situ} measurements and satellite-derived
  estimates of fine particulate matter. The resulting set-indexed Gaussian
  process has higher estimated leave-one-out predictive accuracy, although the
  difference is uncertain. It is about eight times faster for model fitting and
  nine times faster for prediction than the aggregated model.

\end{abstract}
\doublespacing{}

\section{Introduction}

Many modern spatial analyses require defining dependence across observations
indexed by spatial units with heterogeneous locations, sizes, and
geometries. Several closely related statistical problems arise from this
mismatch. Change of support refers to predicting a random field~(RF) indexed on
one spatial scale using observations collected on
another~\citep{gelfand2001change}, with downscaling denoting the special case in
which information aggregated over coarse spatial units is used to infer values
at a finer resolution~\citep{zheng2025spatial}. When observations on the same
phenomenon are simultaneously available at multiple spatial resolutions, the
problem is commonly framed as spatial data
fusion~\citep{moraga2017geostatistical}. In addition, response and explanatory
variables are often observed on incompatible spatial supports, a setting known
as spatial misalignment~\citep{godoy2026voronoi}. Although these settings are
often treated separately, they expose a common methodological challenge:
constructing valid dependence structures on collections of spatial sets rather
than on points.

Most existing approaches to multi-resolution spatial data model dependence by
treating observations at different supports as aggregations of a common
underlying RF, which leads to substantial computational and theoretical
limitations. Specifically, areal measurements are typically assumed to originate
from integrals of a point-level RF~\citep{fuentes2005model}. However, the
implied finite-dimensional probability density functions involve stochastic
integrals that are analytically intractable even when the underlying RF is a
Gaussian Process~(GP). Consequently, inferences rely on numerical approximations
or Monte Carlo integration~\citep{gelfand2001change}, which are computationally
demanding even for moderately sized datasets and introduce hard-to-quantify
bias~\citep{goncalves2018exact}. An alternative strategy avoids aggregation by
defining stochastic processes directly on collections of sets. The Hausdorff
distance has been explored as a basis for such
constructions~\citep{godoy2026statistical}, but the methodology does not
guarantee positive semidefinite~(PSD) covariance functions, leaving the
existence and validity of the induced process unresolved. This lack of
theoretical guarantee, coupled with high computational
costs~\citep{knauer2011directed}, is the motivation for a more robust framework
for set-to-set dependence.

Constructing valid PSD functions becomes substantially more delicate once the
underlying index space departs from Euclidean geometry. While many PSD functions
are known for~\(\mathbb{R}^p\) equipped with the Euclidean
distance~\citep[Ch.~2.5.1]{cressie1993statistics}, their PSD property may not
extend to other spaces~\citep{li2023inference}. In fact, a function that is PSD
on~\(\mathbb{R}\) may not remain PSD on~\(\mathbb{R}^2\), although the reverse is
true~\citep[p.~22]{cressie1993statistics}. This geometric sensitivity has
motivated recent work on characterizing PSD covariance functions on
non-Euclidean spaces. The particular case of the sphere has been extensively
studied~(see \citet{porcu2021nonparametric} and references therein). Certain
covariance functions exhibit the notable property of being PSD on both the
\((p-1)\)-dimensional sphere~(\(\mathbb{S}^{p - 1}\)) and \(p\)-dimensional
Euclidean space~\citep{guinness2016isotropic, porcu2016spatio}. More recently,
analogous questions have arisen for covariance functions indexed by graphs and
their edges~\citep{anderes2020isotropic}. Despite these advances, there is no
general construction of distances on spatial sets that both ensures positive
semi-definiteness of isotropic covariance functions and avoids the computational
burdens associated with aggregation-based models.

We close this gap with a framework for isotropic covariance functions defined
directly on collections of spatial sets. The framework accommodates
point-referenced, areal, and mixed spatial data in a single formulation and
avoids the analytical and computational intractability of aggregation-based
models. For an admissible class of bounded sets with unique minimum enclosing
closed balls~(MEBs), we define the \emph{center-radius~(CR) distance}, a
pseudometric equal to the distance between MEB centers plus the absolute
difference of their radii. In any geodesic space, this distance is a tractable
upper bound for the Hausdorff distance between MEBs. The two distances coincide
across the entire space for complete, locally compact, strongly geodesically
complete length spaces, including Euclidean spaces and Hadamard
manifolds~\citep{barrera2021geodesic}. On the sphere, equality holds when the
center separation plus the larger enclosing radius does not exceed half the
great-circle circumference~(see Remark~\ref{rem:cr_haus_sphere}). Moreover, the
CR distance is conditionally negative definite~(CND) whenever the domain metric
is CND, whereas we show that the Hausdorff distance does not necessarily inherit
this property. The CND property is essential because, by Schoenberg's
theorem~\citep{schoenberg1938metric}, a square-root transformed CND pseudometric
embeds isometrically into a Hilbert space, thereby inducing a rich class of
valid isotropic covariance functions indexed by spatial sets.

The remainder of the paper is organized as follows. Section~\ref{sec:notation}
introduces geodesic spaces, CND functions, and MEBs, and gives conditions for
the uniqueness of MEBs on~\(\mathbb{R}^p\) and~\(\mathbb{S}^{p-1}\).
Section~\ref{sec:cnd} defines the CR distance, derives its properties, and
compares it with the Hausdorff distance. Section~\ref{sec:examples} measures the
agreement between the two distances and the speedup of the CR distance
on~\(\mathbb{R}^2\) and~\(\mathbb{S}^2\). Section~\ref{sec:main} presents our main
theoretical contribution: valid isotropic covariance functions for sets.
Section~\ref{sec:ill} applies the framework to the fusion of point and areal
PM\textsubscript{2.5} data and compares it with an aggregation-based model.
Section~\ref{sec:disc} discusses the results, extensions, limitations, and open
questions. Appendix~\ref{app:proofs} contains the proofs, and
Appendix~\ref{app:agg} gives additional results and convergence checks for
Section~\ref{sec:ill}.

\section{Preliminaries}\label{sec:notation}

Let~\(D\) denote a non-empty index set and~\(d : D \times D \to [0, \infty)\) a distance
function. The pair~\((D, d)\) is a pseudometric space if~\(d\) satisfies all
metric axioms except the identity of indiscernibles, meaning~\(d(x, y) = 0\)
does not necessarily imply that~\(x = y\). We denote closed balls of
radius~\(r \geq 0\) by~\(\B(x; r) = \{ u \in D : d(x, u) \leq r \}\). We follow the
convention of defining the distance between a point~\(x \in D\) and a non-empty
set~\(A \subseteq D\) as the infimum distance between them~(i.e.,
\(d(x, A) = \inf_{a \in A} d(x, a)\)). For $r \ge 0$, we define the
$r$-neighborhood of a set~$A \subseteq D$ as
$\B(A; r) = \bigcup_{a \in A} \B(a; r)$. This union need not be closed.

Our comparison between the CR and Hausdorff distances in Section~\ref{sec:cnd}
assumes that~\((D, d)\) is a geodesic space, a special case of a length space. A
path in~\(D\) is a continuous map~\(\gamma : [a, b] \to D\), and its length is
\[
  L(\gamma) = \sup \sum_{i = 1}^{N} d(\gamma(t_{i - 1}), \gamma(t_{i})),
\]
where the supremum runs over all
partitions~\(a = t_0 \le t_1 \le \cdots \le t_N =
b\)~\citep[Definition~2.3.1]{burago2001course}. The space~\((D, d)\) is a length
space if the distance between any two points equals the infimum of the lengths
of all paths joining them. A length space is geodesic if, for every pair of
points, there exists a shortest path whose length attains this
infimum~\citep[Ch.~I.3]{bridson1999metric}. Every geodesic space is a length
space, but the converse does not hold in general, and shortest paths need not be
unique~\citep[Ch.~I.3]{bridson1999metric}. Euclidean space~\(\mathbb{R}^{p}\) and the
sphere~\(\mathbb{S}^{p-1}\)~(\(p \ge 2\), equipped with its intrinsic metric) are
important examples of geodesic spaces.

A RF on~\(D\) is a collection of random
variables~\(\{Z(x) : x \in D\}\) defined on a common probability
space. Constructing valid probability laws on general pseudometric spaces is
substantially more challenging than on Euclidean domains. Without loss of
generality, we assume that~\(\E[Z(x)] = 0\) for all~\(x \in D\). Under this
assumption, a tractable specification is an isotropic RF. That is, its
covariance function depends solely on the distance between elements of the index
set:
\[
  \mathrm{cov}(Z(x), Z(y)) = C(d(x, y)), \, x, y \in D.
\]
The validity of the random field depends on the isotropic covariance
function~\(C(\cdot)\) being PSD, ensuring that every linear combination of
finite-dimensional realizations has a non-negative variance. An isotropic RF,
such that~\(\mathrm{var}(Z(s)) = \sigma^2 < \infty\) is said second-order~(or weakly)
stationary. An isotropic and second-order stationary GP is a RF such that its
finite-dimensional realizations follow a multivariate Normal distribution.

Our strategy for deriving PSD isotropic covariance functions on sets relies on
the CND property of the underlying distance function. A function
\(g : D \times D \to [0, \infty)\) is CND if it satisfies
\[
  \sum_{i=1}^{m} \sum_{j=1}^{m} b_i b_j g(s_i, s_j) \le 0,
\]
for any points~\(s_1, \dots, s_m \in D\) and real numbers~\(b_1, \dots, b_m\)
satisfying~\(\sum_{i=1}^{m} b_i = 0\). In contrast, a continuous
function~\(\varphi : [0, \infty) \to \mathbb{R}\) is radial PSD on the pseudometric
space~\((D, d)\) if
\[
  \sum_{i=1}^{m} \sum_{j=1}^{m} a_i a_j \varphi(d(s_i, s_j)) \ge 0,
\]
for all choices of~\(s_1, \dots, s_m \in D\) and real
coefficients~\(a_1, \dots, a_m \in \mathbb{R}\). Isotropic covariance
functions~\(C(\cdot)\) represent an important special case of radial PSD functions
on~\((D, d)\).

MEBs provide low-dimensional geometric summaries for sets, reducing bounded sets
to representative location and magnitude proxies. Let~\(\mathcal{F}(D)\) denote the class
of non-empty bounded subsets of~\(D\). For any set \(A \in \mathcal{F}(D)\) and
point~\(x \in D\), the radius of the smallest ball centered at~\(x\)
containing~\(A\) is
\[
  {\rm R}(x, A) = \inf \{ r \ge 0 : A \subseteq \B(x; r) \}.
\]
Constructing a MEB requires selecting a center that minimizes this radius. We
adopt the Chebyshev center~\citep[p.~128]{garkavi1970theory}, which provides an
optimal notion of set centrality by minimizing the supremum of distances to all
points within~\(A\).

\begin{definition}[Chebyshev center]\label{def:cc}
  The \emph{Chebyshev radius} of any
  set~\(A \in \mathcal{F}(D)\), denoted by~\({\rm R}(A)\), is the infimum of
  the radii of all balls in~\(D\) containing~\(A\):
  \[
    {\rm R}(A) = \inf_{x \in D} {\rm R}(x, A).
  \]
  The set of \emph{Chebyshev centers} of~\(A\), denoted by~\(\mathcal{E}(A)\), consists of
  the points in~\(D\) that achieve this minimum radius:
  \[
    \mathcal{E}(A) = \{ x \in D : {\rm R}(x, A) = {\rm R}(A) \}.
  \]
  If~\(\mathcal{E}(A)\) is non-empty, any element~\(c(A) \in \mathcal{E}(A)\) is referred to as a
  Chebyshev center of~\(A\).
\end{definition}

\begin{figure}[tb]
  \centering
  \includegraphics{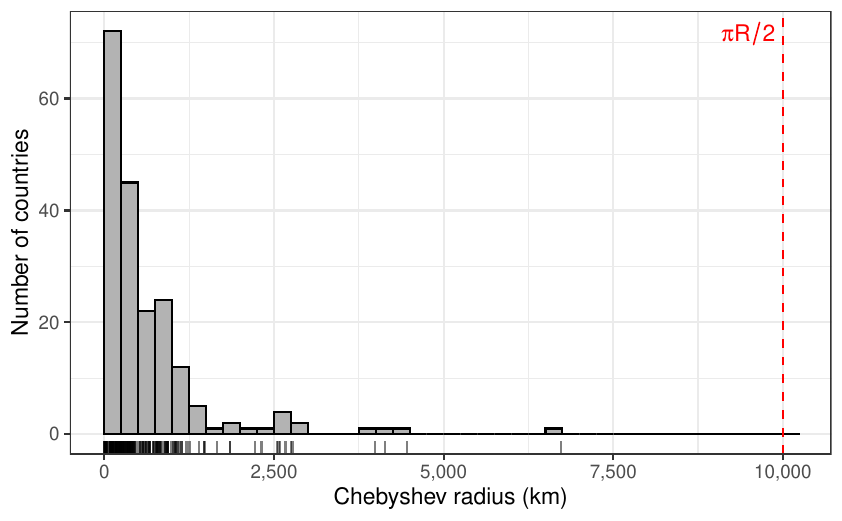}
  \caption{Chebyshev radii of the 242 countries and territories in the Natural
    Earth 1:50 million map, computed on the sphere of radius~\(r =
    6{,}371\)~km. The dashed line marks~\(r\pi/2 \approx 10{,}008\)~km. All radii are
    below this bound, so each country has a unique Chebyshev center.}
  \label{fig:countries}
\end{figure}

The existence and uniqueness of Chebyshev centers depend on the geometry of the
underlying space. We focus on Euclidean spaces~\(\mathbb{R}^{p}\) and the
sphere~\(\mathbb{S}^{p-1}\) and defer further discussions to
Section~\ref{sec:disc}. In Euclidean spaces, the strict convexity of the norm
guarantees that the set of centers~\(\mathcal{E}(A)\) is a
singleton~\citep{amir1978chebyshev}. In contrast, antipodal symmetry can produce
multiple centers on the sphere. For a sphere of radius~\(r\), a Chebyshev
radius~\({\rm R}(A) < r \pi / 2\) suffices for
uniqueness~\citep[Ch.~2]{bridson1999metric}. Figure~\ref{fig:countries} presents
a histogram of the Chebyshev radii associated to 242 countries and
territories. In the figure, it is evident that no country has a Chebyshev radius
greater than~\(r \pi / 2\). Thus, in practice, the condition for uniqueness on the
sphere is not restrictive. Under these conditions, we denote the MEB of any
bounded set~\(A \in \mathcal{F}(D)\) by~\(\B(c(A); {\rm R}(A))\). Finally, Welzl's algorithm
allows for efficiently finding Chebyshev's centers on Euclidean and spherical
domains~\citep{welzl1991smallest}.

\section{On CND distance functions for sets}\label{sec:cnd}

To address the general problem of constructing valid isotropic covariance
functions on collections of sets, rather than on point-indexed domains, we first
need to establish how to quantify the distance between sets. A natural candidate
is the Hausdorff distance, which provides a natural notion of discrepancy
between spatial sets. For non-empty sets~\(A_1, A_2 \subseteq D\), the Hausdorff
distance is defined as
\begin{equation}\label{eq:hausdist}
  d_{\mathrm{h}}(A_1, A_2) = \inf \{ \varepsilon \geq 0 \, : \, A_1 \subseteq \B(A_2; \varepsilon), A_2 \subseteq \B(A_1; \varepsilon) \}.
\end{equation}
If no such~\(\varepsilon\) exists, set~\(d_{\mathrm{h}}(A_1, A_2) := \infty\).
While~\(d_{\mathrm{h}}\) is a metric on the class of non-empty, closed, and
bounded subsets of~\(D\)~\citep{arutyunov2016some} and a pseudometric on the
power set~\(2^{D}\)~\citep[Proposition 7.3.3]{burago2001course}, it presents
significant challenges for statistical modeling. Despite its intuitive
interpretation in spatial statistics~\citep{godoy2022testing,
  godoy2026statistical}, computing the Hausdorff distance is often prohibitively
expensive~\citep{knauer2011directed} and, to the best of our knowledge,
constructing valid PSD functions based directly on~\(d_{\mathrm{h}}\) remains an
open problem.

Verifying the CND property directly for arbitrary set distances is not a
straightforward task, making inheritance from a CND base metric a desirable
property. However, the Hausdorff distance~\(d_{\mathrm{h}}\) fails to preserve
CND property even when~\(d\) is CND. The following remarks give explicit
counterexamples on the plane~\(\mathbb{R}^{2}\) with the Euclidean distance and on the
unit sphere~\(\mathbb{S}^{2}\) with the great-circle distance. Consequently, the
direct application of the Hausdorff distance is not useful under our strategy.
\begin{remark}[Failure of CND for the Hausdorff distance on
  \(\mathbb{R}^{2}\)]\label{rem:haus_cnd}
  Let~\(\mathbb{R}^{2}\) be equipped with the Euclidean
  distance~\(d(x, y) = \|x - y\|_{2}\). The induced Hausdorff
  distance~\(d_{\mathrm{h}}\) is not conditionally negative definite on the
  space of non-empty compact subsets of~\(\mathbb{R}^{2}\).
\end{remark}
\begin{remark}[Failure of CND for the Hausdorff distance on
  \(\mathbb{S}^{2}\)]\label{rem:haus_cnd_s2}
  Let~\(\mathbb{S}^{2} = \{x \in \mathbb{R}^{3} : \|x\|_{2} = 1\}\) be equipped with the
  great-circle distance~\(d(x, y) = \arccos(x^{\top} y)\). Although~\(d\) is
  conditionally negative definite, the induced Hausdorff
  distance~\(d_{\mathrm{h}}\) is not conditionally negative definite on the
  space of non-empty subsets of~\(\mathbb{S}^{2}\).
\end{remark}
See Appendix~\ref{app:proof_haus_cnd} and~\ref{app:proof_haus_cnd_s2} for the
proofs, respectively.

The failure of the Hausdorff distance to preserve the CND property motivates
seeking an alternative set distance that retains clear geometric
interpretability while remaining analytically tractable. Building on the MEB
representation in Definition~\ref{def:cc}, we summarize each bounded set by its
Chebyshev center and Chebyshev radius. The center minimizes the supremum of
distances to points in the set, and the radius measures its spatial extent in
the base metric. These summaries are defined for both points and areal supports
and recover the original location and zero radius for a singleton. Their use
amounts to modeling dependence through support location and extent, with shape,
orientation, and internal geometry omitted once the summaries are fixed. This
representation is appropriate when those omitted features are not needed to
describe dependence for the intended application. Combining center separation
with the absolute radius difference gives the following center-radius~(CR)
distance.
\begin{definition}[Center-radius~(CR) distance]\label{def:crd}
  Let~\((D, d)\) denote a metric space. For non-empty bounded
  sets~\(A_1, A_2 \in \mathcal{F}(D)\) with unique Chebyshev centers, the
  \emph{center-radius distance} is defined as
  \[
    d_{\mathrm{cr}}(A_1, A_2) = d(c(A_1), c(A_2)) + \lvert {\rm R}(A_1) - {\rm R}(A_2)
    \rvert,
  \]
  where~\(c(\cdot)\) and~\({\rm R}(\cdot)\) denote the Chebyshev center and radius,
  respectively.
\end{definition}

The CR distance inherits all pseudometric axioms directly from the underlying
metric space~\((D, d)\). Non-negativity~(\(d_{\mathrm{cr}}(A_1, A_2) \geq 0\)),
reflexivity~($d_{\mathrm{cr}}(A, A) = 0$), and symmetry follow immediately from
the properties of~\(d\) and the absolute difference between radii, while the
triangle inequality follows by summing the triangle inequalities of~\(d\)
on~\(D\) and~\(\lvert \cdot \rvert\) on~\(\mathbb{R}\). Because
\(d_{\mathrm{cr}}(A_1, A_2) = 0\) if and only if~\(c(A_1) = c(A_2)\)
and~\({\rm R}(A_1) = {\rm R}(A_2)\), the function~\(d_{\mathrm{cr}}\) is a
strict metric on the collection of MEBs, and a pseudometric
on~\(\mathcal{F}(D)\). The latter allows distinct sets that share an identical MEB to have
distance zero.  Furthermore, for singletons~\(A_1 = \{x_1\}\)
and~\(A_2 = \{x_2\}\), we
have~\(d_{\mathrm{cr}}(\{x_1\}, \{x_2\}) = d(x_1, x_2)\), whereas for a
singleton~\(\{x\}\) and a set~\(A\) it
yields~\(d_{\mathrm{cr}}(\{x\}, A) = d(x, c(A)) + {\rm R}(A)\).

The additive structure in Definition~\ref{def:crd} preserves conditional
negative definiteness by standard closure properties. The center term is the
pullback of the base distance under the center map, and the radius term is the
pullback of the absolute difference on the real line. When the base distance is
CND, the first term is CND. Because the second term is always CND, so is their
sum. The following theorem records this validity guarantee for the proposed
summaries.
\begin{theorem}\label{th:cr_cnd}
  Let~\((D, d)\) be a metric space where the distance~\(d\) is CND.\@ Then,
  \(d_{\mathrm{cr}}(\cdot, \cdot)\) is CND on~\(\mathcal{F}(D)\) with unique Chebyshev centers.
\end{theorem}
See Appendix~\ref{app:proof_cr_cnd} for the proof.

The Hausdorff geometry of enclosing balls motivates the CR construction. In
Euclidean spaces, the sum of center separation and absolute radius difference is
exactly the Hausdorff distance between the MEBs. In general geodesic spaces, the
CR distance provides an upper bound on the Hausdorff distance between the
corresponding MEBs, and the two quantities coincide in complete, locally
compact, strongly geodesically complete length spaces, such as Euclidean spaces
and Hadamard manifolds~\citep{barrera2021geodesic}. Here, strong geodesic
completeness means that every minimizing geodesic segment can be extended to a
minimizing geodesic ray whose restriction to every bounded interval is distance
realizing. The following proposition formalizes this relationship.
\begin{proposition}[Relationship between CR and Hausdorff
  distances]\label{prop:cr_haus}
  Let~\((D, d)\) be a geodesic space. For any
  sets~$A_1, A_2 \in \mathcal{F}(D)$ with unique Chebyshev centers, the
  Hausdorff distance between their MEBs satisfies
  \begin{equation}\label{eq:cr_haus_ineq}
    d_{\mathrm{h}}(\B(c(A_1); {\rm R}(A_1)), \B(c(A_2); {\rm R}(A_2)))
    \le d_{\mathrm{cr}}(A_1, A_2).
  \end{equation}
  Furthermore, if~\((D, d)\) is a complete, locally compact, strongly
  geodesically complete length space~\citep{barrera2021geodesic}, then equality
  holds in~\eqref{eq:cr_haus_ineq}.
\end{proposition}
See Appendix~\ref{app:proof_cr_haus} for the proof.

\begin{remark}[Spherical equality condition]\label{rem:cr_haus_sphere}
  Let~\(D\) be a sphere of radius~\(a > 0\) and dimension at least one, equipped
  with the great-circle distance~\(d\). For non-empty bounded
  sets~\(A_1, A_2 \subseteq D\) with unique Chebyshev centers, equality
  in~\eqref{eq:cr_haus_ineq} holds if and only if
  \begin{equation}\label{eq:cr_haus_sphere_condition}
    d(c(A_1), c(A_2))
    + \max\{{\rm R}(A_1), {\rm R}(A_2)\} \le a\pi.
  \end{equation}
  To see this, scale the distances in Lemma~\ref{lm:haus_sphere_inner} from the
  unit sphere to~\(D\).  The Hausdorff distance between the MEBs is then the
  minimum of \(a\pi - \min\{{\rm R}(A_1), {\rm R}(A_2)\}\) and
  \(d_{\mathrm{cr}}(A_1, A_2)\). The latter is the minimum precisely
  when~\eqref{eq:cr_haus_sphere_condition} holds. This condition concerns both
  center separation and enclosing radii. Although enclosing radii
  below~\(a\pi/2\) ensure unique centers, that radius bound alone does not ensure
  equality of the CR and Hausdorff distances between the MEBs.
\end{remark}

% The result below is trivial...
% When sets coincide with their own minimum enclosing balls in a Euclidean space,
% Proposition~\ref{prop:cr_haus} implies that the CR and Hausdorff distances are
% identical. In particular, because every compact interval in~\(\mathbb{R}^1\) is a closed
% ball with center at its midpoint and radius equal to its half-width, the two
% distances coincide across all closed intervals on the real line.
% \begin{corollary}\label{co:r1}
%   For any closed intervals~\([a, b], [c, d] \subset \mathbb{R}^1\)
%   with~\(a \le b\) and~\(c \le d\),
%   \(d_{\mathrm{cr}}([a, b], [c, d]) = d_{\mathrm{h}}([a, b], [c, d])\).
% \end{corollary}
% \begin{proof}
%   We have
%   \begin{align*}
%     d_{\mathrm{cr}}([a, b], [c, d])
%     & = \frac{1}{2} [ \lvert (a + b) - (c + d) \rvert
%       + \lvert (b - a) - (d - c) \rvert ] \\
%     & = \frac{1}{2} [ \lvert (b - d) + (a - c) \rvert
%       + \lvert (b - d) - (a - c) \rvert ] \\
%     & = \max (\lvert b - d \rvert, \lvert a - c \rvert).
%   \end{align*}
%   Similarly,
%   \begin{align*}
%     d_{\mathrm{h}}([a, b], [c, d])
%     & = \inf \{ r \geq 0 \, : \, [a, b] \subseteq \B([c, d]; r), \,
%       [c, d] \subseteq \B([a, b]; r) \} \\
%     & = \inf \{ r \geq 0 \, : \, [a, b] \subseteq [c - r, d + r], \,
%       [c, d] \subseteq [a - r, b + r] \} \\
%     & = \max (\lvert b - d \rvert, \lvert a - c \rvert).
%   \end{align*}
%   Therefore, \(d_{\mathrm{cr}}([a, b], [c, d]) = d_{\mathrm{h}}([a, b], [c, d])\), as claimed.
% \end{proof}

\subsection{Empirical agreement between the CR and Hausdorff
  distances}\label{sec:examples}

\begin{figure}[tb]
  \centering
  \includegraphics{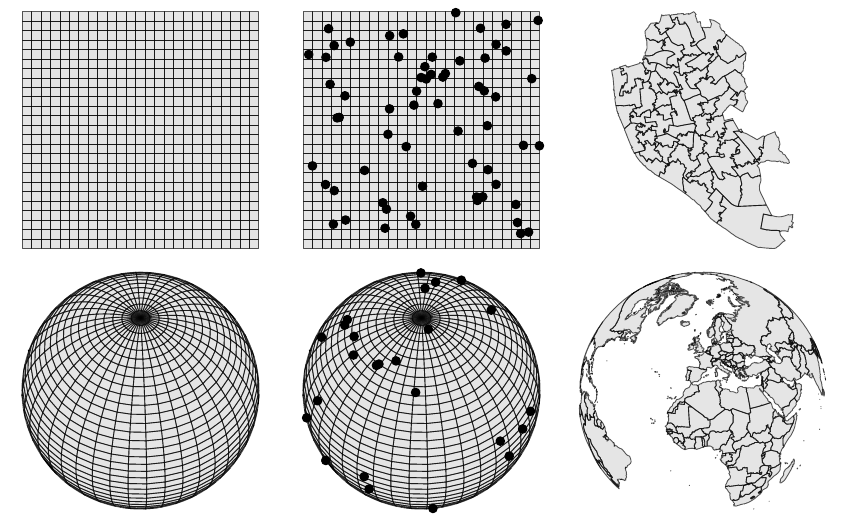}
  \caption{Synthetic spatial configurations used in the numerical examples.
    Rows correspond to the Euclidean plane~\((\mathbb{R}^2)\) and the
    sphere~\((\mathbb{S}^2)\), while columns represent gridded, fused, and areal
    settings. These examples are designed to span increasing spatial complexity
    across domains commonly encountered in multi-resolution spatial analyses.}
  \label{fig:maps}
\end{figure}

To empirically evaluate the proposed distance, we consider synthetic datasets
designed to reflect three common spatial configurations encountered in practice:
gridded data~(typical of satellite-derived products), fused data~(combining
point-referenced and areal observations), and areal data consisting of irregular
polygons. These configurations are evaluated on both the Euclidean
plane~\((\mathbb{R}^2)\) and the sphere~\((\mathbb{S}^2)\). Using these settings, our
analysis pursues two objectives. First, to assess empirical agreement between
the proposed CR distance and the Hausdorff distance; second, to benchmark the
computational efficiency of the proposed distance's closed-form
representation. The empirical agreement analysis is motivated by the Hausdorff
distance being a natural candidate for a distance between sets with intuitive
interpretation in spatial statistics.

\begin{figure}[tb]
  \centering
  \includegraphics{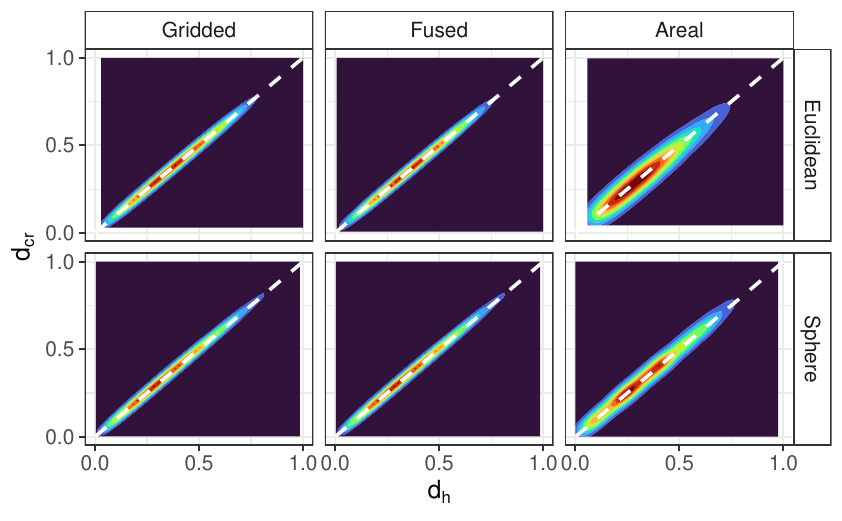}
  \caption{Two-dimensional kernel density estimates comparing pairwise Hausdorff
    distances~(x-axis) and CR distances~(y-axis) across spatial configurations
    and domains. Columns correspond to gridded, fused, and areal settings, while
    rows represent the Euclidean plane and the sphere. Distances are rescaled to
    the unit interval within each scenario. Concentration of mass along the
    dashed~\(y = x\) line indicates close empirical agreement between the two
    distance measures across increasing spatial complexity.}
  \label{fig:pairs}
\end{figure}

We computed pairwise distance matrices for all configurations shown in
Figure~\ref{fig:maps} using both the CR and Hausdorff distances. Each distance
matrix was rescaled to the unit interval by dividing by the maximum observed
distance within each scenario, so that comparisons reflect relative rather than
absolute scale differences across spatial configurations. We computed Pearson
and Spearman correlation coefficients between the upper-triangular elements of
the distance matrices. The Pearson correlations ranged from 0.997~(areal data on
\(\mathbb{R}^2\)) to 1~(lattice on~\(\mathbb{R}^2\)), while the Spearman correlations ranged from
0.996~(areal on~$\mathbb{R}^2$) to 1~(areal data on~$\mathbb{R}^2$). Figure~\ref{fig:pairs}
displays two-dimensional kernel density estimates comparing the resulting
distance values. The concentration of mass along the~\(y = x\) diagonal further
supports the close empirical agreement between the two distances, with no
systematic curvature or dispersion as spatial complexity increases from gridded
to fused and irregular areal units. In particular, the near-monotone
relationship across all configurations suggests that the CR distance preserves
the relative proximity structure induced by the Hausdorff distance, despite
relying on a substantially reduced geometric representation.

\begin{table}[htp]
  \centering
  \caption{Ratio of Hausdorff to CR distance computation times on the Euclidean
    plane~(\(\mathbb{R}^2\)) and sphere~(\(\mathbb{S}^2\)) across data
    types~(Setting). Columns report the ratios of the average~(Avg.),
    median~(Med.), minimum~(Min.), and maximum~(Max.) computation times across
    50 iterations.}
  \label{tab:times}
  \begin{tabular}{llcccc}
    \toprule
    Space & Setting & Avg. & Med. & Min. & Max.\\ 
    \midrule
    \(\mathbb{R}^2\)          & Lattice & 278.00 & 257.40 & 351.45 & 241.79 \\ 
                     &   Fused & 240.39 & 236.49 & 313.88 & 208.70 \\ 
                     &   Areal & 130.81 & 143.66 & 153.71 &  33.87 \\ 
    \(\mathbb{S}^2\) & Lattice &  82.54 &  81.75 &  90.22 &  79.57 \\ 
                     &   Fused &  68.43 &  70.15 &  76.72 &  48.51 \\ 
                     &   Areal & 591.22 & 611.60 & 647.50 & 275.28 \\ 
    \bottomrule
  \end{tabular}
\end{table}

The two distances differ substantially in computational cost. Evaluation of the
Hausdorff distance requires identifying extremal discrepancies between spatial
sets, which becomes increasingly expensive as set complexity and resolution
increase. In contrast, the CR distance admits a closed-form expression that
depends only on Chebyshev centers and radii. This structural reduction leads to
substantial computational savings. In our benchmark, using the \texttt{R}
package \texttt{bench}~\citep{hester2025bench}, the evaluation based on the CR
distance was between~33.9 and~647.5 times faster than evaluation based on the
Hausdorff distance, with the largest speedups observed for areal data~(see
Table~\ref{tab:times}). We make code used to evaluate both distances available
in the \texttt{smile} \texttt{R} package~\citep{Rpkg:smile}.

\section{Valid isotropic covariance models for sets}\label{sec:main}

Isometric embeddings are distance-preserving mappings between spaces that help
in linking CND distance functions to valid isotropic covariance
models. Formally, a pseudometric space~\((D, d)\) can be isometrically embedded
in a target space~\((D^\ast, d^\ast)\) if there exists a
map~\(\phi \, : \, D \to D^\ast\) such that
\[
  d^\ast(\phi(x), \phi(y)) = d(x, y),
\]
for any~\(x, y \in D\)~\citep[Ch.~1]{wells1975embeddings}. More generally,
\((D, d)\) is~\(f\)-embeddable in~\((D^\ast, d^\ast)\) if there exists a
map~\(\phi \, : \, D \to D^\ast\) and a function~\(f \, : \, [0, \infty) \to [0, \infty)\) such that
\[
  d^\ast(\phi(x), \phi(y)) = f(d(x, y)),
\]
for any~\(x, y \in D\)~\citep{tang2024space}. When~\(f(x) = x^{1/2}\), the
space~\((D, d)\) is said to be square-root embeddable in~\((D^\ast, d^\ast)\). A
classical result establishes an equivalence between conditional negative
definiteness and isometric embeddings into Hilbert
spaces~\citep{schoenberg1938metric}.
\begin{theorem}\label{th:schoenberg}
  A pseudometric space~\((D, d)\) can be isometrically embedded in a Hilbert
  space~\(H\) if and only if~\(d^2\) is CND.\@
\end{theorem}
Proofs~\citep[Theorem~2.4 and Remark~3.2]{wells1975embeddings} and further
discussion~\citep{menegatto2020gneiting} of Theorem~\ref{th:schoenberg} are
available in the literature.

Hilbert-space embeddings allow isotropic covariance families defined on
Euclidean domains to be transferred to general pseudometric spaces equipped with
CND distances. Denote~\(\Phi_p\) as the class of radial PSD functions
on~\(\mathbb{R}^p\). Functions belonging to these classes satisfy the inclusion
relations~\citep{gneiting2013strictly}:
\[
  \Phi_1 \supset \Phi_2 \supset \cdots \supset \Phi_{\infty}
  = \bigcap_{k = 1}^{\infty} \Phi_{k},
\]
where~\(\Phi_{\infty}\) corresponds to the class of radial PSD functions on arbitrary
Hilbert spaces~\citep{gneiting2013strictly}. Well-known members of the
class~\(\Phi_{\infty}\) include the Mat\'{e}rn, powered exponential~(PEXP), generalized
Cauchy, and Dagum families~\citep[Table~2]{gneiting2013strictly}.  Combining
Theorem~\ref{th:schoenberg} with the definition of~\(\Phi_{\infty}\) provides sufficient
conditions for constructing valid isotropic covariance models on non-Euclidean
domains equipped with CND pseudometrics, including spaces based on
\(\alpha\)-norms~\citep{richards1985positive}, great-circle
distances~\citep{meckes2013positive}, and graph-indexed
domains~\citep{anderes2020isotropic}.
\begin{theorem}\label{th:2}
  Let~\((D, d)\) be a pseudometric space such that~\(d\) is CND.\@ Then, any
  function belonging to the class~\(\Phi_{\infty}\) is radial PSD on~\((D, d^{1/2})\).
\end{theorem}
See Appendix~\ref{app:proof_th2} for the proof.

Combining the CND inheritance of the CR distance in Theorem~\ref{th:cr_cnd} with
Theorem~\ref{th:2} establishes valid isotropic covariance functions for sets. In
particular, whenever a metric space~\((D, d)\) is equipped with a CND distance
and admits unique Chebyshev centers, the square-root transformed CR distance
embeds isometrically into a Hilbert space. Consequently, covariance families
in~\(\Phi_{\infty}\) remain valid on collections of bounded sets after a square-root
transformation of~\(d_{\mathrm{cr}}\), while the PEXP family with power
parameter~\(\nu \in (0, 1]\) is valid directly on the original CR distance
scale. The following theorem and corollary formalize these two properties for
set-indexed processes.
\begin{theorem}\label{th:main}
  Let~\((D, d)\) be a metric space such that~\(d\) is CND and bounded sets
  in~\(\mathcal{F}(D)\) admit unique Chebyshev centers. Then,
  \((\mathcal{F}(D), d_{\mathrm{cr}})\) is square-root embeddable in a Hilbert space,
  and any function belonging to the class~\(\Phi_{\infty}\) is radial PSD
  on~\((\mathcal{F}(D), d_{\mathrm{cr}}^{1/2})\).
\end{theorem}

\begin{corollary}\label{co:main}
  Let~\((\mathcal{F}(D), d_{\mathrm{cr}})\) be a pseudometric space such
  that~\(d_{\mathrm{cr}}\) is CND.\@ Then, \(\exp\{ -\kappa d_{\mathrm{cr}}^{\nu} \}\) is PSD
  on~\(\mathcal{F}(D)\) for~\(\nu \in (0, 1]\) and~\(\kappa > 0\).
\end{corollary}
See Appendix~\ref{app:proof_th_main} and~\ref{app:proof_co_main} for the
proofs of Theorem~\ref{th:main} and Corollary~\ref{co:main}, respectively.

\section{Illustration with Real Data}\label{sec:ill}

We illustrate the application of the set-indexed covariance functions to spatial
data fusion problem for fine particulate matter~(PM\textsubscript{2.5}). The
illustration compares a set-indexed GP based on the CR distance with a
traditional aggregated GP in terms of predictive accuracy and computing time
for model fitting and prediction. We used a dataset from previous
analyses~\citep{moraga2017geostatistical, godoy2026statistical}, which
combines~19 \textit{in situ} measurements~\citep{epa2022clean} with~183 areal
satellite-derived estimates~\citep{van2015global}, averaged from 2010 to 2012,
in the Ventura and Los Angeles counties, California, USA\@. The main goal is a
change of support where we used both data sources to predict
PM\textsubscript{2.5} on a grid finer than the satellite product.

The statistical models we consider are hierarchical and share the same
observation model. Let~\(D = \mathbb{R}^2\) be the ambient space, equipped with the
Euclidean distance~\(d\). Sampling units are contained in the study region, and
their MEB centers are defined by minimizing over all of~\(D\).  For a sampling
unit~\(A \in \mathcal{F}(D)\), let~\(Y(A)\) denote the observed PM\textsubscript{2.5}
measurement, and let~\(W(A)\) be a binary indicator that equals~\(1\) for
satellite-derived areal estimates and~\(0\) for \textit{in situ} point
measurements. We modeled~\(Y(A)\) as
\[
  (Y(A) \mid \boldsymbol{\beta}, Z(A), \tau^2) \sim
  \mathcal{N}(\beta_0 + \beta_1 W(A) + (1 + \beta_2 W(A)) Z(A), \tau^2),
\]
where~\(Z(\cdot)\) is a zero-mean latent GP indexed by sets, \(\beta_0\) is an
intercept, \(\tau^2\)~denotes a nugget effect, and~\(\beta_1\) and~\(\beta_2\) are additive
and multiplicative bias terms for the satellite-derived
data~\citep{fuentes2005model}. Given~\(Z(\cdot)\), the observations
\(Y(A_1), \ldots, Y(A_N)\) are conditionally independent. The two models differ only
in the specification of~\(Z(\cdot)\).

In the proposed model, \(Z(\cdot)\) is a set-indexed zero-mean isotropic GP
on~\(\mathcal{F}(D)\) with the PEXP covariance function with power~\(\nu = 1\):
\[
  \mathrm{cov}(Z(A_1), Z(A_2)) = \sigma^2 \exp\left\{ -\frac{d_{\mathrm{cr}}(A_1, A_2)}{\xi}
  \right\}, \quad A_1, A_2 \in \mathcal{F}(D),
\]
where~\(\sigma^2\) is the marginal variance and~\(\xi > 0\) controls the range of the
spatial dependence. This model is valid by construction. The Euclidean distance
is CND~\citep{schoenberg1938metric}, and bounded sets in \(\mathbb{R}^2\) have unique
Chebyshev centers~(Section~\ref{sec:notation}). Hence, \(d_{\mathrm{cr}}\) is
CND by Theorem~\ref{th:cr_cnd}, and the covariance function is PSD by
Corollary~\ref{co:main}. Evaluating this covariance function requires only the
Chebyshev center and radius of each set.

To the best of our knowledge, no other construction yields isotropic and
second-order stationary GPs indexed by sets. We therefore compared the proposed
model with the traditional aggregated GP, which averages a point-indexed GP over
each set. The supports considered here are positive-area units or singletons:
\begin{equation}\label{eq:integral}
  Z(A_i)
  = \begin{cases}
    {\mathcal{A}(A_i)}^{-1}
    \int_{A_i} Z(x) \dd x, & \text{if }
    \mathcal{A}(A_i) > 0, \\
    Z(x_i), & \text{if } A_i = \{x_i\},
  \end{cases}
\end{equation}
where~\(\mathcal{A}(A_i)\) denotes the area of~\(A_i\)
and~\(\{ Z(x) \, : \, x \in D \}\) is a point-level zero-mean, second-order
stationary GP with an isotropic PEXP covariance function with power~\(\nu =
1\). The finite-dimensional distributions of this model are multivariate Normal,
but the covariances that involve areal units require a numerical approximation
of the integrals~\citep[for further details, see][]{gelfand2001change}.

The two models coincide at the point level. For singletons,
\(d_{\mathrm{cr}}(\{x\}, \{y\}) = d(x, y)\), and~\eqref{eq:integral}
gives~\(Z(\{x\}) = Z(x)\).  Hence, both models assign the
covariance~\(\sigma^2 \exp\{ -d(x, y) / \xi \}\) to any two point observations, and
with point data alone the two models would be identical. They differ only in the
covariances that involve areal units. The set-indexed model evaluates the PEXP
function at the CR distance, which has a closed form. The aggregated model
averages the point-level covariance over the areas, so each such covariance
requires a numerical integral. Any differences in fit and computing time
reported below therefore arise from how each model treats areal supports.

In both models, marginalizing~\(\mathbf{Z} = [Z(A_1), \ldots, Z(A_N)]^{\top}\) gives a
multivariate Normal distribution
for~\(\mathbf{Y} = [Y(A_1), \ldots, Y(A_N)]^\top\) with
mean~\(\boldsymbol{\mu} = \mathbf{X} [\beta_0, \,\, \beta_1]^\top\), where
the~\(i\)-th row of~\(\mathbf{X}\) is~\([1, \,\, W(A_i)]\), and covariance
matrix
\[
  \Lambda = (\mathbf{I}_N + \beta_2 \mathbf{W}) \Sigma
  (\mathbf{I}_N + \beta_2 \mathbf{W}) + \tau^2 \mathbf{I}_N.
\]
Here, \(\Sigma_{ij} = \mathrm{cov}(Z(A_i), Z(A_j))\) comes from either the
set-indexed or the aggregated model, \(\mathbf{W}\) is a diagonal matrix with
entries~\(W(A_i)\), and~\(\mathbf{I}_N\) is the~\(N \times N\) identity matrix. We
fitted both models in a Bayesian framework with the same prior distributions,
using the No-U-Turn sampler~\citep[NUTS]{hoffman2014no} in
\texttt{Stan}~\citep{stan2023rstan}. Appendix~\ref{app:agg} gives the priors,
the sampler settings, the convergence diagnostics, and the parameter estimates.

We compared the candidate models using the paired difference in their
leave-one-out cross-validation~(LOO-CV) estimates of the expected log pointwise
predictive density~(\(\mathrm{elpd}_{\mathrm{LOO}}\)) and its associated
uncertainty~\citep{vehtari2017practical, sivula2025uncertainty}, alongside each
model's effective number of parameters~\(p_{\mathrm{loo}}\) and computing time
for model fitting and prediction. Rather than refitting each model separately
for all \(N\) held-out observations, we approximated the LOO-CV
estimator~\(\widehat{\mathrm{elpd}}_{\mathrm{LOO}}\) via Pareto smoothed
importance sampling~(PSIS) using posterior draws from the full-data fit as the
proposal distribution~\citep{vehtari2024pareto}. For each
observation~\(i = 1, \ldots, N\), PSIS yields an estimated Pareto shape
parameter~\(k_i\) that diagnoses the reliability of the importance-sampling
approximation. With the 4000 posterior draws used here, we used~$k_i = 0.7$ as
the diagnostic threshold. For the observation above this threshold, we refit the
model omitting that sampling unit and estimated
its~\(\mathrm{elpd}_{\mathrm{LOO}}\) contribution from the refitted posterior
directly~\citep{sivula2025uncertainty}. We then quantified the uncertainty in
the resulting paired~\(\widehat{\mathrm{elpd}}_{\mathrm{LOO}}\) differences
using the corresponding standard error~\citep{sivula2025uncertainty}.

\begin{table}[htp]
\centering
\caption{Computing time in seconds for the full-data model fit~(Fit), excluding
  LOO refits, and
  prediction~(Pred.), paired difference in expected log predictive density
  \(\Delta\widehat{\mathrm{elpd}}_{\mathrm{LOO}}\) relative to the set-indexed model
  (with standard errors in parentheses), and effective number of parameters
  \(p_{\mathrm{loo}}\) (with standard errors in parentheses).}
\label{tab:fit}
\begin{tabular}{lrrrr}
  \toprule
  Model & Fit (s)
  & Pred.\ (s)
  & \(\Delta\widehat{\mathrm{elpd}}_{\mathrm{LOO}}\) (SE)
  & \(p_{\mathrm{loo}}\) (SE) \\ 
  \midrule
  Set-indexed & 104.76 & 15.00 & -- & 6.47 (2.96) \\ 
  Aggregated & 876.32 & 136.53 & $-$52.8 (30.3) & 17.29 (4.72) \\
  \bottomrule
\end{tabular}
\end{table}

The set-indexed model achieves a computational speedup over its aggregated
counterpart while delivering higher estimated out-of-sample predictive
accuracy~(Table~\ref{tab:fit}). Specifically, the set-indexed model requires
104.76~s for fitting and 15.00~s for prediction, compared with 876.32~s and
136.53~s for the aggregated model. Relative to the set-indexed baseline, the
aggregated model yields a paired expected log predictive density difference of
\(\Delta\widehat{\mathrm{elpd}}_{\mathrm{LOO}} = -52.8\)~(SE:~\(30.3\)). While all
Pareto shape parameters~\(k_i\) for the set-indexed model remain below~0.7, the
aggregated model has a single observation with~\(k_i = 0.82\). We refitted the
aggregated model without this observation to estimate its LOO contribution.
This refit required an additional 807.91~s for predictive assessment, separate
from the full-data fitting times in Table~\ref{tab:fit}. Furthermore, the
effective number of parameters is \(p_{\mathrm{loo}} =
6.47\)~(\(\mathrm{SE} = 2.96\)) for the set-indexed model, whereas
\(p_{\mathrm{loo}} = 17.29\)~(\(\mathrm{SE} = 4.72\)) for the aggregated
model. Because the paired difference lies within two standard errors, the
predictive gain of the set-indexed model is suggestive rather than conclusive.

\begin{figure}[tb]
  \centering
  \includegraphics{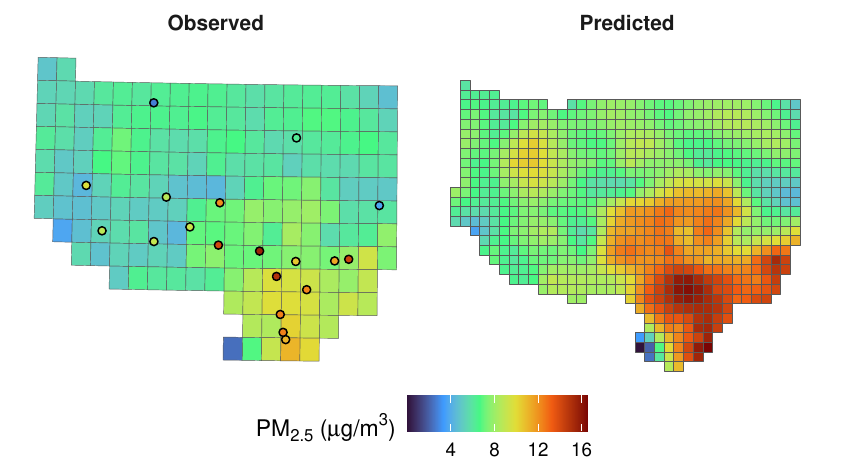}
  \caption{On the left, the observed PM\textsubscript{2.5} data. On the right,
    the change of support based on the median of the posterior predictive
    distribution of~\(\mathbf{Y}^\ast\) under the set-indexed model.}
  \label{fig:pm25}
\end{figure}

Both models can perform the change of support through the posterior predictive
distribution. We predicted PM\textsubscript{2.5} at 743 grid
cells~\([A^{\ast}_1, \ldots, A^{\ast}_{743}]^{\top}\), which have a finer resolution than the
satellite-derived estimates. Because the satellite-derived estimates are biased,
we set~\(W(A^{\ast}_i) = 0\) for all grid cells, so that the predictions exclude
the additive and multiplicative biases of the satellite data. Standard results
for the multivariate Normal distribution give Monte Carlo samples from the
posterior predictive distribution
of~\(\mathbf{Y}^{\ast} = [Y(A^{\ast}_1), \ldots, Y(A^{\ast}_{743})]^{\top}\) from the posterior
samples of the parameters~\citep{godoy2026statistical}. Figure~\ref{fig:pm25}
shows the observed data and the posterior predictive median
of~\(\mathbf{Y}^{\ast}\) according to the set-indexed GP. In summary, the
set-indexed model is valid by construction and has higher estimated LOO
predictive accuracy, although the difference is uncertain. It is about eight
times faster for full-data fitting and nine times faster for prediction than the
aggregated model, while the two models coincide at the point level.

\section{Discussion}\label{sec:disc}

We introduced the CR distance, a new distance function for sets, and established
its theoretical properties in set-indexed settings. For sets with unique
Chebyshev centers, the proposed distance admits a closed-form expression that
depends only on their Chebyshev centers and radii. This structural reduction
allows the CR distance to inherit conditional negative definiteness from the
base metric, which in turn guarantees the validity of broad classes of
isotropic covariance functions through Hilbert-space square-root embeddings. We
demonstrate that the Hausdorff distance does not generally inherit the CND
property of the base metric. The contribution is a covariance modeling
framework based on MEB summaries, supported by standard CND closure and
embedding results. The geometric comparisons explain the choice of summaries
and their combination, while the validity guarantee makes the resulting
models available for point, areal, and mixed supports.

The closed-form expression induced by the geometric reduction of the CR distance
leads directly to substantial computational simplification. By representing each
set through its center and radius, evaluation of the distance avoids the
extremal point-to-point calculations required by the Hausdorff distance. The
resulting computational gains observed in our numerical studies, including
speedups of over 33 times in areal data configurations, follow from this
structural simplification rather than from approximation or numerical
shortcuts. The real-data illustration in Section~\ref{sec:ill} compares the
set-indexed GP with the traditional aggregated GP. Both models coincide at the
point level, but the aggregated model needs one numerical integral for each pair
of sets that involves an areal unit. For the PM\textsubscript{2.5} data, the
set-indexed model has a higher estimated LOO expected log predictive density,
with a paired difference of 52.8~(SE:~30.3) in its favor. It is about eight
times faster for full-data fitting and nine times faster for prediction; the
aggregated model also required an additional LOO refit.

The covariance family determines the scale at which a model uses the CR
distance. While the PEXP family with~\(\nu \in (0,1]\) can be applied directly to
the CR distance, Theorem~\ref{th:main} guarantees validity of the
Mat\'ern class on the square-root transformed distance. This transformation
is sufficient for validity; the theorem does not establish its necessity
for every parameter choice or index domain. Because the
square-root is a monotone transformation, inference on range-related parameters
remains available on the original distance scale. In frequentist settings,
invariance of the maximum likelihood estimator permits reporting the practical
range on the original scale, with uncertainty quantified, for example, via the
delta method. In Bayesian analyses, the same transformation may be applied to
posterior samples or handled explicitly through a Jacobian adjustment if one
desires to place a prior on the range parameter on the original scale.

The validity results hold beyond our examples and beyond the Chebyshev choice of
center and radius. Theorems~\ref{th:cr_cnd} and~\ref{th:main} need only a CND
distance~\(d\) and at least one Chebyshev center for each set. If a set has
several Chebyshev centers, both theorems hold for any fixed choice of center,
but~\(d_{\mathrm{cr}}\) then depends on that choice. The same proofs apply
to~\(\lambda_{c} d(\tilde{c}(A_1), \tilde{c}(A_2)) + \lambda_{R} \lvert \tilde{R}(A_1) -
\tilde{R}(A_2) \rvert\), where~\(\tilde{c} : \mathcal{F}(D) \to D\)
and~\(\tilde{R} : \mathcal{F}(D) \to \mathbb{R}\) are any center and magnitude maps,
and~\(\lambda_{c}, \lambda_{R} \ge 0\) weight the center and magnitude terms. The CR
construction uses Chebyshev summaries and unit coefficients to preserve the
Hausdorff bound and equality for MEBs~(Proposition~\ref{prop:cr_haus}). Other
summaries or coefficients preserve CND but need not retain this
relationship. Thus, the geometric interpretation motivates our choice without
requiring it for covariance validity. The benchmarks and application support
this construction for the studied settings; they do not establish that its
summaries or coefficients are optimal for other applications.

Separate covariance factors provide another way to allow different effects of
location and radius:
\[
  K(A_1, A_2) = C_1(d(c(A_1), c(A_2))) \, C_2(\lvert {\rm R}(A_1) - {\rm R}(A_2)
  \rvert), \quad A_1, A_2 \in \mathcal{F}(D),
\]
where~\(C_1\) and~\(C_2\) are valid isotropic covariance functions on~\((D, d)\)
and~\(\mathbb{R}\), respectively. Both factors are PSD on~\(\mathcal{F}(D)\), and so is their
product~\citep[Theorem~4.4~(b)]{wells1975embeddings}. In general, \(K\) is not
isotropic in~\(d_{\mathrm{cr}}\), but it allows separate ranges for location and
magnitude. When both factors use the powered exponential family with power
parameter one, their product is an exponential covariance based on a nonnegative
weighted sum of center separation and absolute radius difference. With equal
ranges and a product of marginal variances equal to~\(\sigma^2\), this gives the
covariance in Section~\ref{sec:ill}. This reduction need not hold for other
choices of~\(C_1\) and~\(C_2\).

The reduction to centers and radii places a substantive restriction on the
latent process. For a zero-mean GP whose covariance is a function of the CR
distance, sets with identical centers and radii have zero variance for their
difference, and hence equal latent values almost surely. Thus, shape and
orientation cannot affect their latent values or their covariances with other
sets. An independent nugget with positive variance makes finite observation
covariance matrices nonsingular, but does not restore dependence on the omitted
geometry. Applications in which shape or internal spatial variation affects
dependence require additional summaries or another covariance construction. In
particular, the proposed model specifies dependence directly across supports and
does not impose the integral relationships of an aggregated point-indexed
process. Its suitability should therefore be assessed against the scientific
interpretation of the observations as well as predictive fit.

Several theoretical questions remain open. One concerns the mean-square
differentiability of random fields defined through the proposed covariance
construction. For the standard Mat\'ern family, differentiability is governed by
the smoothness parameter~\citep[pg.~32]{stein1999interpolation}, but extending
this notion to set-indexed or manifold-indexed processes requires more general
concepts of differentiation~\citep{khastan2021new}. This issue is non-trivial,
even for Gaussian processes on spheres~\citep{alegria2021f}, and while recent
work has defined mean-square differentiability for Gaussian processes on
manifolds using smooth paths~\citep{lin2019extrinsic}, a systematic treatment in
the present setting remains to be developed. Another interesting question is
whether one may construct valid covariance functions in~\(\mathcal{F}(D)\) without the
reliance on the square-root embedding strategy.

% ::BEGIN_ACK::

\section*{Acknowledgments}

Lucas da Cunha Godoy was partially supported by the Walton Family Foundation
(grant \#00109884). Fernando A. Quintana was partially funded by grant ANID
Fondecyt Regular 1260012. Marcos O. Prates acknowledges the research grants
obtained from CNPq-Brazil (318023/2025-3), FAPEMIG (APQ-01837-22, APQ-01748-24),
and CAPES, respectively, for partial financial support.

% ::END_ACK::

\bibliographystyle{apalike}
\bibliography{refs}
\appendix

\section{Proofs}\label{app:proofs}

\subsection*{Closed balls in geodesic spaces}

The proofs of Lemma~\ref{lm:haus_sphere_inner} and
Proposition~\ref{prop:cr_haus} in Appendix~\ref{app:proofs} rely on the linear
expansion of closed balls in geodesic spaces. The expansion identity requires a
shortest path between any two points, which exists by definition in a geodesic
space~\citep[Ch.~2]{burago2001course}. By the Hopf--Rinow--Cohn-Vossen
theorem~\citep[Theorem~2.5.28]{burago2001course}, every complete, locally
compact length space is geodesic, including Euclidean spaces~\(\mathbb{R}^{p}\) and the
sphere~\(\mathbb{S}^{p-1}\). The following lemma therefore extends the
corresponding result for complete, locally compact length
spaces~\citep{barrera2021geodesic} to general geodesic spaces.

\begin{lemma}[Linear expansion of balls]\label{lm:closed_ballball}
  Let~\((D, d)\) be a geodesic space. For any~\(x \in D\) and
  radii~\(r, k \ge 0\),
  \[
    \B(\B(x; k); r) = \B(x; r + k).
  \]
\end{lemma}
\begin{proof}
  Let~\(z \in \B(\B(x; k); r)\). Then there is a point~\(y \in \B(x; k)\)
  such that~\(d(y, z) \le r\). The triangle inequality gives
  \(d(x, z) \le d(x, y) + d(y, z) \le k + r\), so
  \(z \in \B(x; r + k)\).

  Conversely, let~\(z \in \B(x; r + k)\), and let~\(\gamma : [0, L] \to D\)
  be a minimizing geodesic from~\(x\) to~\(z\), parameterized by arc length.
  Thus~\(L = d(x, z) \le r + k\). Set~\(t = \min\{k, L\}\) and
  \(y = \gamma(t)\). Then~\(d(x, y) = t \le k\), and
  \(d(y, z) = L - t = \max\{L - k, 0\} \le r\). Hence~\(y \in \B(x; k)\)
  and~\(z \in \B(y; r)\), so~\(z \in \B(\B(x; k); r)\).
  The two inclusions prove the identity.
\end{proof}

Applying Lemma~\ref{lm:closed_ballball} reduces the \(\varepsilon\)-neighborhood
of a closed ball to a concentric ball of expanded radius, yielding an explicit
upper bound on the Hausdorff distance between closed balls.
\begin{lemma}[Hausdorff distance between balls in geodesic
  spaces]\label{lm:haus_reduction}
  Let~\((D, d)\) be a geodesic space. For any~\(x, y \in D\) and
  radii~\(r, s \ge 0\),
  \begin{equation}\label{eq:haus_balls_inf}
    d_{\mathrm{h}}(\B(x; r), \B(y; s)) = \inf \left\{ \varepsilon \ge 0 \, : \,
      \B(x; r) \subseteq \B(y; s + \varepsilon), \,
      \B(y; s) \subseteq \B(x; r + \varepsilon) \right\},
  \end{equation}
  and the following upper bound holds:
  \begin{equation}\label{eq:haus_upper}
    d_{\mathrm{h}}(\B(x; r), \B(y; s)) \le d(x, y) + \lvert r - s \rvert.
  \end{equation}
\end{lemma}
\begin{proof}
  For any~\(\varepsilon \ge 0\), Lemma~\ref{lm:closed_ballball} gives
  \(\B(\B(y; s); \varepsilon) = \B(y; s + \varepsilon)\) and
  \(\B(\B(x; r); \varepsilon) = \B(x; r + \varepsilon)\). Substituting these
  identities into~\eqref{eq:hausdist} yields~\eqref{eq:haus_balls_inf}.
  
  To establish~\eqref{eq:haus_upper}, set
  \(\varepsilon = d(x, y) + \lvert r - s \rvert\). For any~\(u \in \B(x; r)\), the triangle
  inequality implies
  \[
    d(y, u) \le d(y, x) + d(x, u) \le d(x, y) + r \le s + \varepsilon,
  \]
  so~\(\B(x; r) \subseteq \B(y; s + \varepsilon)\). By symmetry, for
  any~\(v \in \B(y; s)\), we have
  \(d(x, v) \le d(x, y) + s \le r + \varepsilon\), which implies
  \(\B(y; s) \subseteq \B(x; r + \varepsilon)\). Thus~\(\varepsilon\) is an
  admissible value in~\eqref{eq:haus_balls_inf}, establishing~\eqref{eq:haus_upper}.
\end{proof}

\subsection{Remark~\ref{rem:haus_cnd}}\label{app:proof_haus_cnd}

Consider five compact subsets of~\(\mathbb{R}^{2}\) placed along the horizontal axis:
\begin{align*}
  A_{1} & = \{(1, 0)\}, \quad
          A_{2} = \{(2, 0)\}, \quad
          A_{3} = \{(0, 0), (2, 0)\}, \\
  A_{4} & = \{(1, 0), (3, 0)\}, \quad
          A_{5} = \{(0, 0), (1, 0), (2, 0)\}.
\end{align*}
Direct computation of the pairwise Hausdorff distances
\(\Delta_{ij} = d_{\mathrm{h}}(A_{i}, A_{j})\) yields the~\(5 \times 5\) matrix
\[
  \Delta =
  \begin{pmatrix}
    0 & 1 & 1 & 2 & 1 \\
    1 & 0 & 2 & 1 & 2 \\
    1 & 2 & 0 & 1 & 1 \\
    2 & 1 & 1 & 0 & 1 \\
    1 & 2 & 1 & 1 & 0
  \end{pmatrix}.
\]
For the vector~\(b = (-2, 2, 1, -2, 1)^{\top} \in \mathbb{R}^{5}\), which
satisfies~\(\sum_{i=1}^{5} b_{i} = 0\), we obtain
\[
  b^{\top} \Delta b = 2 > 0,
\]
violating the conditional negative definiteness inequality.

\subsection{Remark~\ref{rem:haus_cnd_s2}}\label{app:proof_haus_cnd_s2}

To evaluate the Hausdorff distance between spherical balls explicitly, we first
establish a closed-form expression on the unit sphere~\(\mathbb{S}^{p}\), whose
diameter is~\(\operatorname{diam}(\mathbb{S}^{p}) = \pi\).
\begin{lemma}[Exact Hausdorff distance between balls on
  \(\mathbb{S}^{p}\)]\label{lm:haus_sphere_inner}
  For any~\(x, y \in \mathbb{S}^{p}\)~(\(p \ge 1\)) and
  radii~\(r, s \in [0, \pi]\),
  \begin{equation}\label{eq:haus_sphere_inner}
    d_{\mathrm{h}}\left(\B(x; r), \B(y; s)\right)
    = \min\left\{ \pi - \min(r, s), \, d(x, y) + |r - s| \right\}.
  \end{equation}
\end{lemma}
\begin{proof}[Proof of Lemma~\ref{lm:haus_sphere_inner}]
  Without loss of generality, assume~\(r \ge s\), so that~\(\min(r, s) = s\)
  and~\(|r - s| = r - s\). Let~\(M = \min\{ \pi - s, \, d(x, y) + r - s \}\).

  By~\eqref{eq:haus_upper} in Lemma~\ref{lm:haus_reduction}, we have
  \(d_{\mathrm{h}}(\B(x; r), \B(y; s)) \le d(x, y) + r - s\). Moreover, for any
  \(\varepsilon > \pi - s\), we have
  \(r + \varepsilon \ge s + \varepsilon > \pi\). Since
  \(\operatorname{diam}(\mathbb{S}^{p}) = \pi\), any ball of radius strictly
  greater than~\(\pi\) equals~\(\mathbb{S}^{p}\), so the inclusions
  \(\B(x; r) \subseteq \B(y; s + \varepsilon)\) and
  \(\B(y; s) \subseteq \B(x; r + \varepsilon)\) in~\eqref{eq:haus_balls_inf}
  hold trivially. Thus~\(d_{\mathrm{h}}(\B(x; r), \B(y; s)) \le \pi - s\), establishing
  \(d_{\mathrm{h}}(\B(x; r), \B(y; s)) \le M\).

  For the reverse inequality, let~\(\varepsilon \ge 0\) satisfy
  \(\B(x; r) \subseteq \B(y; s + \varepsilon)\). Let
  \(\gamma : [0, \pi] \to \mathbb{S}^{p}\) be a minimizing great-circle geodesic
  parameterized by arc length from~\(\gamma(0) = y\) through
  \(\gamma(d(x, y)) = x\) to~\(\gamma(\pi) = -y\), and set
  \(t^{*} = \min\{ \pi, \, d(x, y) + r \}\). Since
  \(d(x, y) \le t^{*} \le d(x, y) + r\), the point~\(u = \gamma(t^{*})\)
  satisfies~\(d(x, u) = t^{*} - d(x, y) \le r\), so
  \(u \in \B(x; r) \subseteq \B(y; s + \varepsilon)\). Consequently,
  \(t^{*} = d(y, u) \le s + \varepsilon\), that is,
  \(\varepsilon \ge t^{*} - s = M\). Taking the infimum over~\(\varepsilon\)
  via~\eqref{eq:haus_balls_inf} gives
  \(d_{\mathrm{h}}(\B(x; r), \B(y; s)) \ge M\).
\end{proof}

We now construct the counterexample on~\(\mathbb{S}^{2}\). Let
\(q(\theta) = (\cos \theta, \sin \theta, 0) \in \mathbb{S}^{2}\) denote a point on the equator,
and consider the three centers
\[
  x_{1} = q(\pi), \quad
  x_{2} = q\left(\frac{3\pi}{2}\right), \quad
  x_{3} = q\left(\frac{5\pi}{8}\right),
\]
whose pairwise great-circle distances are~\(d(x_{1}, x_{2}) = \pi / 2\),
\(d(x_{1}, x_{3}) = 3\pi / 8\), and~\(d(x_{2}, x_{3}) = 7\pi / 8\). Define six
spherical balls with radii strictly below~\(\pi / 2\):
\[
  A_{1} = \B(x_{1}; \pi/4), \quad
  A_{2} = \B(x_{2}; \pi/4), \quad
  A_{3} = \B(x_{3}; \pi/4),
\]
\[
  A_{4} = \B(x_{3}; 3\pi/8), \quad
  A_{5} = \B(x_{1}; 3\pi/8), \quad
  A_{6} = \B(x_{2}; 3\pi/8).
\]
Evaluating the pairwise Hausdorff
distances~\(\Delta_{ij} = d_{\mathrm{h}}(A_{i}, A_{j})\) via Lemma~\ref{lm:haus_sphere_inner}
yields the~\(6 \times 6\) matrix
\[
  \Delta = \frac{\pi}{8}
  \begin{pmatrix}
    0 & 4 & 3 & 4 & 1 & 5 \\
    4 & 0 & 6 & 6 & 5 & 1 \\
    3 & 6 & 0 & 1 & 4 & 6 \\
    4 & 6 & 1 & 0 & 3 & 5 \\
    1 & 5 & 4 & 3 & 0 & 4 \\
    5 & 1 & 6 & 5 & 4 & 0
  \end{pmatrix}.
\]
For the vector~\(b = (-2, 1, 1, -1, 2, -1)^{\top} \in \mathbb{R}^{6}\), which
satisfies~\(\sum_{i=1}^{6} b_{i} = 0\), we obtain
\[
  b^{\top} \Delta b = \frac{\pi}{4} > 0.
\]
Thus, the Hausdorff distance is not CND in this context.

\subsection{Theorem~\ref{th:cr_cnd}}\label{app:proof_cr_cnd}

\begin{proof}[Proof of Theorem~\ref{th:cr_cnd}]
  By definition, we have
  \[
    d_{\mathrm{cr}}(A_1, A_2) = d(c(A_1), c(A_2))
    + \lvert {\rm R}(A_1) - {\rm R}(A_2) \rvert,
  \]
  where~\(d(c(A_1), c(A_2))\) is CND on~\(D\) by assumption. The second term,
  \(\lvert {\rm R}(A_1) - {\rm R}(A_2) \rvert\), corresponds to the
  \(\alpha\)-norm on~\(\mathbb{R}\) with~\(\alpha = 1\), which is also
  CND~\citep{richards1985positive}. Because the family of CND functions forms a
  convex cone~\citep[p.~88]{cressie1993statistics}, any linear combination of
  CND functions with non-negative coefficients remains CND. Thus, under the
  stated assumptions, \(d_{\mathrm{cr}}(A_1, A_2)\) is CND.\@
\end{proof}

\subsection{Proposition~\ref{prop:cr_haus}}\label{app:proof_cr_haus}

\begin{proof}[Proof of Proposition~\ref{prop:cr_haus}]
  For any~\(A_1, A_2 \in \mathcal{F}(D)\) with unique
  Chebyshev centers~\(c(A_1)\) and~\(c(A_2)\) and Chebyshev
  radii~\({\rm R}(A_1)\) and~\({\rm R}(A_2)\), applying~\eqref{eq:haus_upper} in
  Lemma~\ref{lm:haus_reduction} with~\(x = c(A_1)\), \(r = {\rm R}(A_1)\),
  \(y = c(A_2)\), and~\(s = {\rm R}(A_2)\) gives
  \[
    d_{\mathrm{h}}(\B(c(A_1); {\rm R}(A_1)), \B(c(A_2); {\rm R}(A_2)))
    \le d(c(A_1), c(A_2)) + \lvert {\rm R}(A_1) - {\rm R}(A_2) \rvert
    = d_{\mathrm{cr}}(A_1, A_2),
  \]
  which establishes~\eqref{eq:cr_haus_ineq}. If~\((D, d)\) is a complete,
  locally compact length space, its closed balls are compact. Under the
  additional assumption of strong geodesic completeness, the Hausdorff
  distance between any two such balls equals the distance between their
  centers plus the absolute difference of their
  radii~\citep{barrera2021geodesic}. Applying this result to the two MEBs
  gives equality in~\eqref{eq:cr_haus_ineq}.
\end{proof}

\subsection{Theorem~\ref{th:2}}\label{app:proof_th2}

\begin{proof}[Proof of Theorem~\ref{th:2}]
  By Theorem~\ref{th:schoenberg}, the space~\((D, d^{1/2})\) is a
  pseudometric space that can be isometrically embedded in a Hilbert
  space~\(H\). Consequently, all functions in the class~\(\Phi_{\infty}\) are radial
  PSD on~\((D, d^{1/2})\).
\end{proof}

\subsection{Theorem~\ref{th:main}}\label{app:proof_th_main}

\begin{proof}[Proof of Theorem~\ref{th:main}]
  By Theorem~\ref{th:cr_cnd}, \(d_{\mathrm{cr}}\) is a CND pseudometric
  on~\(\mathcal{F}(D)\). By Theorem~\ref{th:schoenberg}, its square root
  admits an isometric representation in a Hilbert space. The radial PSD
  property for functions in~$\Phi_{\infty}$ follows from
  Theorem~\ref{th:2}.
\end{proof}

\subsection{Corollary~\ref{co:main}}\label{app:proof_co_main}

\begin{proof}[Proof of Corollary~\ref{co:main}]
  Since~\(d_{\mathrm{cr}}\) is a non-negative CND function, \(d_{\mathrm{cr}}^{\nu}\) is also CND
  for any~\(\nu \in (0, 1]\)~\citep[Corollary~2.10]{wells1975embeddings}. The result
  then follows immediately from Schoenberg's
  theorem~\citep[Theorem~2.4]{wells1975embeddings}, which states that a
  function~\(g\) is CND if and only if~\(\exp\{ -\kappa g \}\) is PSD for
  all~\(\kappa > 0\).
\end{proof}

% ::BEGIN_SUPPLEMENT::

\section{Additional results and convergence checks}\label{app:agg}
\begin{figure}[tb]
  \centering
  \includegraphics[width=\linewidth]{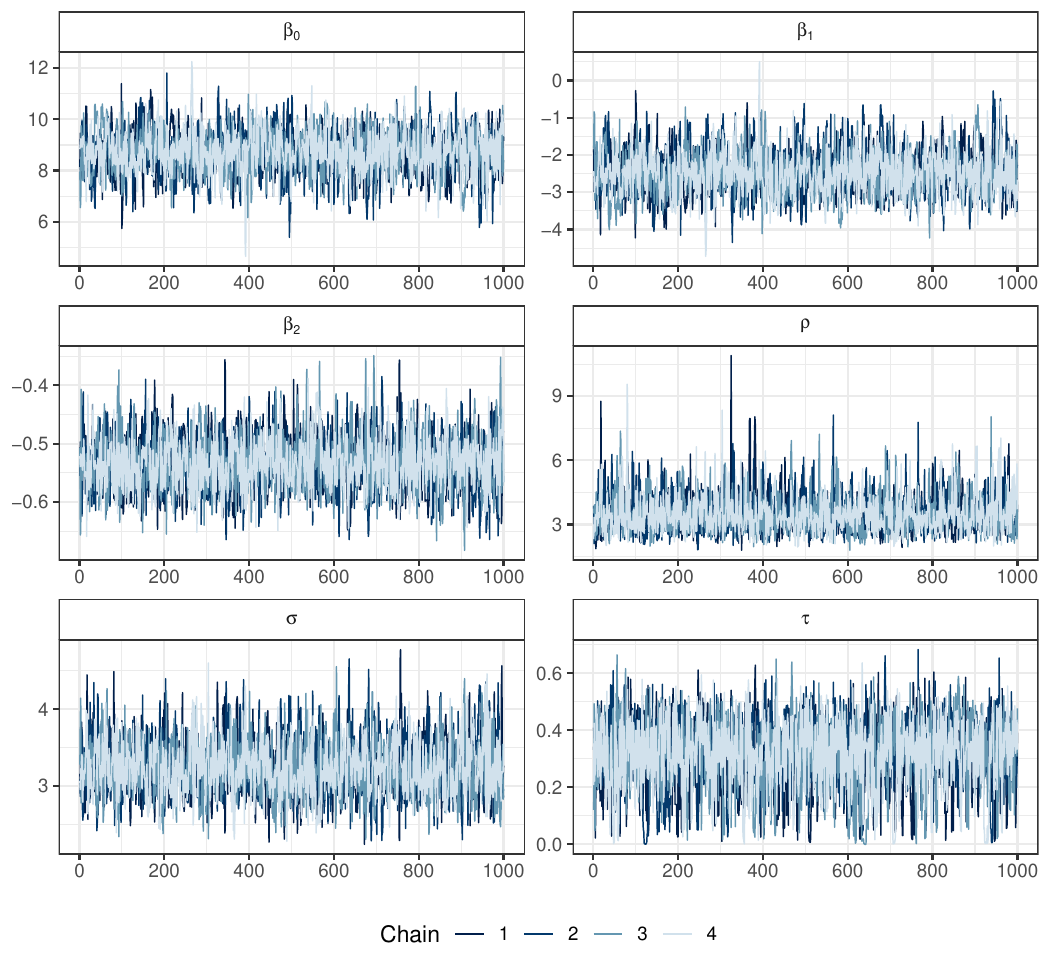}
  \caption{Traceplots of the four MCMC chains for the aggregated model.}
  \label{fig:tagg}
\end{figure}

\begin{figure}[tb]
  \centering
  \includegraphics[width=\linewidth]{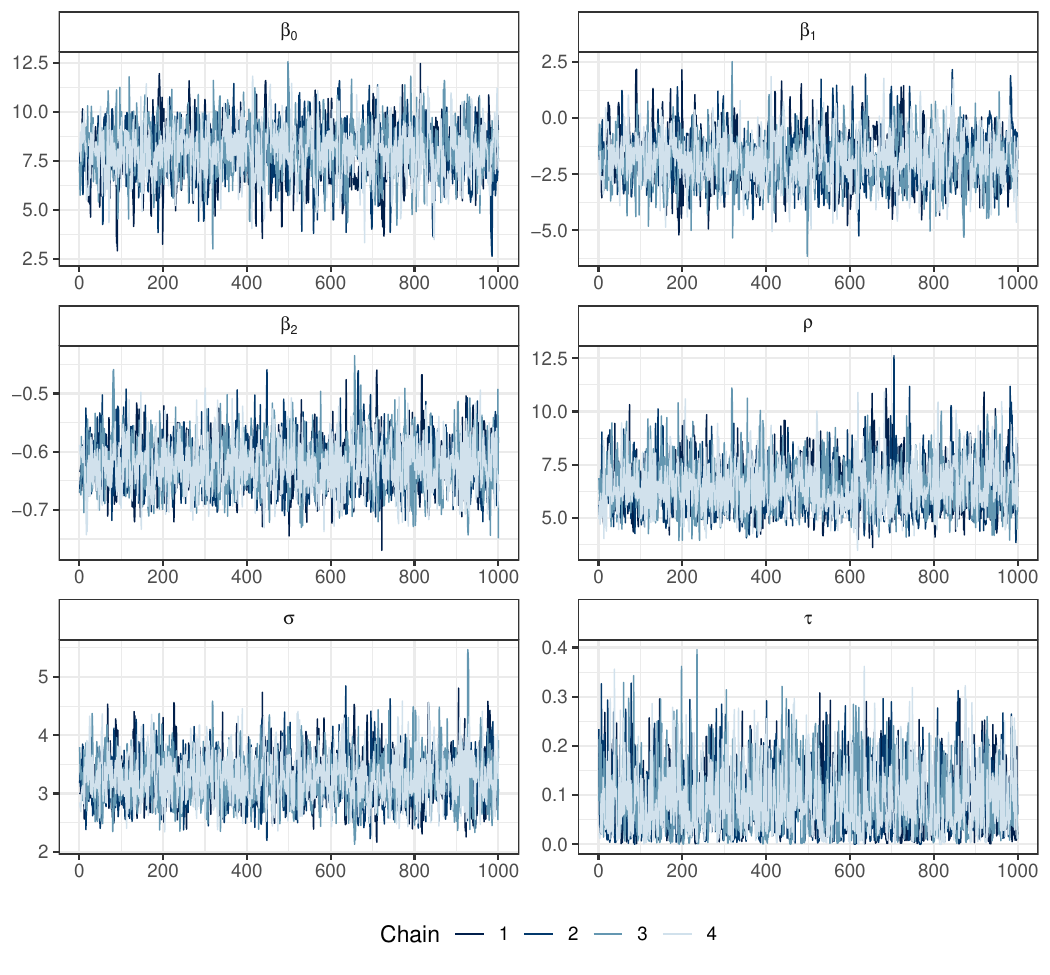}
  \caption{Traceplots of the four MCMC chains for the set-indexed model.}
  \label{fig:tset}
\end{figure}

\begin{table}[htb]
\centering
\caption{Posterior medians and 95\% credible intervals~(Est), split-\(\hat{R}\),
  and bulk and tail effective sample sizes~(ESS) for the aggregated~(Agg) and
  set-indexed~(Set) models. The practical range~\(\rho\) is in units of 10~km.}
\label{tab:params}
\begin{tabular}{lllrrr}
  \toprule
  Parameter & Model & Est & $\hat{R}$ & ESS (Bulk) & ESS (Tail) \\
  \midrule
  $\beta_0$ & Agg & 8.75 [6.94; 10.33] & 1.00 & 1628 & 1695 \\ 
        & Set & 8.02 [5.22; 10.61] & 1.00 & 1532 & 1892 \\ 
  $\beta_1$ & Agg & $-$2.48 [$-$3.56; $-$1.19] & 1.00 & 1534 & 1569 \\
        & Set & $-$1.85 [$-$3.96; 0.49] & 1.00 & 1556 & 1860 \\
  $\beta_2$ & Agg & $-$0.54 [$-$0.62; $-$0.44] & 1.00 & 2005 & 2090 \\
        & Set & $-$0.62 [$-$0.70; $-$0.54] & 1.00 & 2213 & 2470 \\
  $\rho$ & Agg & 3.35 [2.24; 5.68] & 1.00 & 1825 & 2160 \\ 
      & Set & 6.34 [4.55; 9.13] & 1.00 & 2521 & 2390 \\ 
  $\sigma$ & Agg & 3.22 [2.58; 4.05] & 1.00 & 1911 & 1842 \\ 
      & Set & 3.23 [2.53; 4.14] & 1.00 & 1936 & 2202 \\ 
  $\tau$ & Agg & 0.32 [0.03; 0.52] & 1.00 & 1514 & 912 \\ 
      & Set & 0.07 [0.00; 0.24] & 1.00 & 1787 & 1467 \\ 
   \bottomrule
\end{tabular}
\end{table}

This appendix gives the prior distributions, the sampler settings, the
convergence diagnostics, and the parameter estimates for the two models of
Section~\ref{sec:ill}. Both models used the same prior distributions. We
set~\([\beta_{0}, \beta_{1}, \beta_{2}]^{\top} \sim \mathcal{N}(\mathbf{0}, 100 \mathbf{I}_3)\), truncated
standard Normal priors for~\(\sigma\) and~\(\tau\), and an exponential prior for the
practical range~\(\rho = \xi \log(10)\), the distance at which the spatial
correlation drops to~0.1. We set the mean of the exponential prior
as~$\lambda_\rho = -\rho_0 / \log(p_\rho)$, where~$\rho_0$ is set to four-fifths of the largest
observed CR distance in the data, and $p_\rho = 0.05$ is the prior probability
that~$\rho$ exceeds~$\rho_0$~\citep{godoy2026statistical}.  We ran four parallel MCMC
chains, initializing each parameter~(on an unconstrained scale) from a standard
uniform distribution. Each chain had a warm-up period of~1000 iterations, and we
kept the next~1000 iterations for inference. We assessed convergence with
traceplots~(Figures~\ref{fig:tagg} and~\ref{fig:tset}), the split-\(\hat{R}\)
diagnostic~\citep{vehtari2021rank}, effective sample sizes, and No-U-Turn
Sampler~(NUTS) diagnostics. Neither model exhibits any divergent transitions,
maximum treedepth warnings, or other NUTS diagnostic warnings. All
split-\(\hat{R}\) values are below 1.01, and all bulk and tail effective sample
sizes exceed 900~(Table~\ref{tab:params}), supporting convergence.

Table~\ref{tab:params} reports the posterior medians and 95\% credible intervals
of the model parameters. The two models give similar estimates of the
multiplicative bias~\(\beta_2\), with posterior medians of~\(-0.54\) (aggregated)
and~\(-0.62\) (set-indexed). For the additive bias~\(\beta_1\), the posterior
medians are~\(-2.48\) (aggregated) and~\(-1.85\) (set-indexed).  The credible
interval of~\(\beta_1\) under the set-indexed model includes zero, whereas the
interval under the aggregated model does not. The marginal standard
deviations~\(\sigma\) are almost equal under the two models, and the nugget standard
deviation~\(\tau\) is smaller under the set-indexed model~(\(0.07\)
against~\(0.32\)). The set-indexed model also suggests a stronger spatial
dependence: the posterior median of the practical range is 63.4~km, against
33.5~km for the aggregated model.

% ::END_SUPPLEMENT::

\end{document}